\documentclass[conference]{IEEEtran}

\usepackage{graphicx}
\usepackage{amsmath}
\usepackage{amssymb}
\usepackage{amsthm}
\usepackage{dsfont}
\usepackage{subfigure}
\usepackage{tabularx}
\usepackage{algorithm}
\usepackage{algorithmic}
\usepackage{color}
\usepackage{epsfig}
\usepackage{url}

\newtheorem{Theorem}{Theorem}
\newtheorem{Corollary}{Corollary}
\newtheorem{Definition}{Definition}

\usepackage{tabularx}
\usepackage{subfigure}

\begin{document}
	%
	
	\title{Incorporating TSN/BLS in AFDX for Mixed-Criticality Avionics Applications: Specification and Analysis}
	\author{A. FINZI, A. MIFDAOUI, F. FRANCES, E. LOCHIN\\
	University of Toulouse-ISAE, France\\}


		\maketitle
	\begin{abstract}
		In this paper, we propose an extension of the AFDX standard, incorporating a TSN/BLS shaper, to homogenize the avionics communication architecture, and enable the interconnection of different avionics domains with mixed-criticality levels, e.g., legacy AFDX traffic, Flight Control and In-Flight Entertainment. First, we present the main specifications of such a proposed solution. Then, we detail the corresponding worst-case timing analysis, using the Network Calculus framework, to infer real-time guarantees. Finally, we conduct the performance analysis of such a proposal on a realistic AFDX configuration. Results show the efficiency of the Extended AFDX standard to noticeably enhance the medium priority level delay bounds, while respecting the higher priority level constraints, in comparison with the legacy AFDX standard. 
	\end{abstract}
	

	\section{Introduction}
	\label{intro}
	The growing number of interconnected end-systems and the expansion of exchanged data in avionics have led to an increase in complexity of the communication architecture. To cope with this trend, a first communication solution based on a high rate backbone network, i.e., the AFDX (Avionics Full Duplex Switched Ethernet) \cite{ARINC664}, has been implemented by Airbus in the A380, to interconnect critical subsystems. More recently, some low rate data buses, e.g., CAN \cite{CAN}, have been introduced in the A350 and the A400M to handle some specific avionics domains, such as the I/O process and the Flight Control Management. Although this architecture reduces the time to market, it conjointly leads to inherent heterogeneity and new challenges to guarantee the real-time requirements.
	
	To cope with these emerging issues, with the maturity and reliability progress of the AFDX after a decade of successful use, a homogeneous avionic communication architecture based on such a technology to interconnect different avionics domains may bring significant advantages, such as quick installation and maintenance and reduced weight and costs. Furthermore, this new communication architecture needs to support, in addition to the legacy AFDX traffic profile, called Rate Constrained (RC) traffic, at least two extra profiles. The first, denoted by Safety-Critical Traffic (SCT), is specified to support flows with hard real-time constraints and the highest priority level, e.g., flight control data; whereas the second is for Best-Effort (BE) flows with no delivery constraint and the lowest priority level, e.g., In-Flight Entertainment traffic. It is worth noting that such a priority assignment depending on the traffic criticality is usually used in avionics to facilitate the certification process.
	
	Hence, an ultimate avionic communication architecture has to fulfil the following necessary conditions:
	
	\textbf{Condition 1:} keeping the low (re)configuration effort and high modularity level, guaranteed by the current AFDX standard;
	
	\textbf{Condition 2:} supporting mixed-criticality applications while guaranteeing the predictability requirement, i.e., the existence of a bounded latency for each traffic flow;
	
	\textbf{Condition 3:} enforcing the Quality of Service (QoS) features while limiting the impact of the highest priority traffic on the legacy AFDX traffic. 
	
	There are many existing solutions to support mixed-criticality applications in embedded systems, and particularly in avionics \cite{steiner2009ttethernet} \cite{TSN} \cite{Teener2012} \cite{tianran2012design} \cite{Coelho2014}. However, each one satisfies some of the aforementioned conditions better than others, but there is no solution meeting all of them. 
	
	%
	Therefore, our main contributions in this paper are threefold: (i) \textbf{first}, the specification of an extended AFDX standard satisfying the three aforementioned conditions, based on the Bust-Limiting Shaper (BLS) \cite{Gotz2012} defined in Time Sensitive Networking (TSN) task group \cite{TSN}, in \textbf{Section \ref{SpecExtended}}; (ii) \textbf{second}, an appropriate system modeling and timing analysis, based on the Network Calculus framework \cite{leboudecthiran12}, to evaluate its impact on SCT and RC guarantees in \textbf{Sections \ref{TAM}} and \textbf{\ref{SC-NC}}; (iii) \textbf{third}, in \textbf{Section \ref{PA}}, the validation of such a proposal in the case of a realistic avionics network, interconnecting more than 60 end-systems and varying the maximum utilisation rate of SCT and RC traffic, to prove its efficiency to guarantee SCT traffic constraints, while enhancing the RC guarantees. 
	
	\section{Related Work}
	In this section, we review the most relevant solutions to support mixed-criticality applications in embedded systems, and the main worst-case timing analyses of the TSN/BLS shaper.
	
	\subsection{Supporting Mixed-criticality Applications in Embedded Systems }
	Various solutions have been proposed in the literature to support mixed-criticality applications in embedded systems and particularly in avionics \cite{Steiner14}. These solutions can be actually categorised according to the implemented communication paradigm, i.e., mainly event-triggered or time-triggered. This parameter is of utmost importance to quantify the reconfiguration effort needed by the alternative avionics communication architecture, in comparison to the current AFDX standard (which is an event-triggered solution). Furthermore, it conditions the modularity level of the selected solution. The event-triggered paradigm is known as highly flexible and facilitates the system reconfiguration, but it infers at the same time an indeterminism level and needs further proofs to verify the predictability requirement. On the other hand, the time-triggered paradigm is highly predictable, but presents some limitations in terms of system reconfigurability. The considered solutions in this area vs the three specification conditions of an alternative communication architecture, described in Section \ref{intro}, are illustrated in Table \ref{benchmarking}. 
	\begin{table}[h!]
		\footnotesize
		\begin{center}
			\begin{tabular} {|c|c|c|c|c|c|c|c|}
				\hline
				Solutions &  TTE & TAS & PS & BLS & AVB & NP-SP & WRR \\
				\hline
				refs. &   \cite{steiner2009ttethernet} & \cite{thiele2015formal} & \cite{Teener2012} & \cite{Gotz2012} &  \cite{Qav21009} & \cite{Coelho2014}  & \cite{tianran2012design} \\
				\hline
				Cond. 1 & & & & X & X & X &  \\
				\hline
				Cond. 2 & X & X & X & X &  & X & X  \\
				\hline
				Cond. 3 &  &  &  & X & X & & X \\
				\hline
			\end{tabular}
		\end{center}
		\footnotesize \caption{Existing solutions vs specification conditions}
		\label{benchmarking}
	\end{table}

	
	The main relevant solutions implementing time-triggered paradigm on top of Switched Ethernet are Time Triggered Ethernet (TTE) \cite{steiner2009ttethernet} and two solutions under standardisation in IEEE 802 Time Sensitive Networking (TSN) Task group \cite{TSN}, i.e., primarily driven by automotive. Both TSN solutions are based on new shapers, implemented on top of Non-Preemptive Static Priority (NP-SP) scheduler: Time-Aware Shaper (TAS)  \cite{thiele2015formal} and Peristaltic Shaper (PS) \cite{Teener2012} \cite{thiele2015formal}.  All of these solutions aim to drastically reduce the jitter of highest priority traffic, using a time scheduler and/or an adequate synchronisation protocol. These facts imply the non fulfilment of both specification conditions 1 and 3, which make these solutions inadequate in our context.
	
	Among the most interesting solutions based on event-triggered paradigm, we distinguish two classes of solutions. The first class is extending the AFDX standard with well-known scheduling schemes, e.g., NP-SP with n priority levels \cite{Coelho2014} satisfying only conditions 1 and 2 due to the starvation risk for lower priorities, and Weighted-Round-Robin (WRR) \cite{tianran2012design} meeting only conditions 2 and 3 due to the necessary weights setting, known to be difficult to tune. The second class in this category is integrating credit-based shapers to control generally the highest priority level, in order to limit its impact on lower priority ones and to guarantee real-time communication. This idea has been initiated for Ethernet \cite {Kweon00}, then applied to Switched Ethernet \cite{loeser2004low} \cite{mifdaoui2006real}. More recently, it has been integrated in Ethernet AVB \cite{Qav21009}. However, the implemented behaviour in AVB, i.e., the credit becomes negative as soon as one high priority frame is sent, has led to undesirable high latencies for critical traffic. Consequently, this solution satisfies only conditions 1 and 3. On the other hand,  there is an interesting solution in TSN group, based on the Burst Limiting Shaper (BLS)\cite{Gotz2012} on top of NP-SP scheduler \cite{thiele2016formal}, which handles the limitations of AVB standard while keeping the event-triggered paradigm of Switched Ethernet.  These characteristics meet all the specification conditions. Therefore, it is considered herein as the most interesting solution to be incorporated within the AFDX standard, to enable an homogeneous avionics communication architecture. However, the BLS shaper may lead at the same time to increasing communication latencies for the highest priority traffic, thus requiring real-time constraints verification. Therefore, an appropriate timing analysis to provide worst-case delays has to be considered. An overview of the related work concerning this issue is presented in the next section.
	
	\subsection{Worst-case Timing Analysis of TSN/BLS Shaper}
	There are some interesting approaches in the literature concerning the worst-case timing analysis of TSN network, and more particularly BLS shaper. The first and seminal one in \cite{Kerschbaum2013} introduces a first service curve model to induce worst-case delay computation. However, this presentation published by the TSN task group has never been extended in a formal paper. The second one has detailed a more formal worst-case timing analysis in \cite{thangamuthu2015analysis}. However, this approach has some limitations. Basically, the proposed model does not take into account the impact of either the same priority flows or the higher ones, which will clearly induce optimistic worst-case delays. The last and more recent one in \cite{thiele2016formal} has proposed a formal analysis of TSN/BLS shaper, based on a Compositional Performance Analysis (CPA) method. This approach has handled the main limitations of the model presented in \cite{thangamuthu2015analysis}; and interesting results for an automotive case study have been detailed. The impact of BLS on the highest priority traffic has been showed to deteriorate its timing performance, in comparison with a classic NP-SP scheduler.
	
	However, in this paper, our main objective is different from \cite{thiele2016formal} and consists in incorporating BLS in AFDX, denoted as Extended AFDX, to guarantee the highest priority traffic deadline, while limiting its impact on the medium one, i.e., RC. Moreover, our worst-case timing analysis is based on the Network Calculus framework, which has been proved as highly modular and scalable, in comparison with CPA method \cite{Perathoner08}, and very effective to prove the certification requirements of avionics applications \cite{grieu2004analyse}. Several existing works have used Network Calculus to analyse the timing performance of Switched Ethernet and AFDX \cite{grieu2004analyse} \cite{zhao2017timing} \cite{loeser2004low} \cite{fidler2005traffic}. However, to the best of our knowledge, the issue of modeling and analysing the TSN/BLS on top of a NP-SP scheduler using the Network Calculus has not been handled yet in the literature. 
	
	\section{Specification of Extended AFDX }
	\label{SpecExtended}
	In this section, we first describe the Extended AFDX switch architecture, implementing the TSN/BLS on top of a NP-SP scheduler. Then, we detail the BLS behaviour and its main parameters.
	
	\subsection{The Extended AFDX Switch}
	The aim of extending the AFDX switch architecture with the TSN/BLS  is to handle mixed criticality data, and more specifically three AFDX traffic profiles, as illustrated in Fig.\ref{fig:sw_archi}: (i) the Safety-Critical Traffic (SCT) with its priority set by the BLS and the tightest temporal deadline, e.g., Flight-control flows; (ii) Rate Constraint traffic (RC) with the medium priority and a deadline constraint to guarantee, e.g., legacy AFDX flows; (iii) The best-effort traffic (BE) with the lowest priority and no time constraint, e.g., In-Flight Entertainment. 
	
	\begin{figure}[h]
		\centering
		\includegraphics[width=0.40\textwidth]{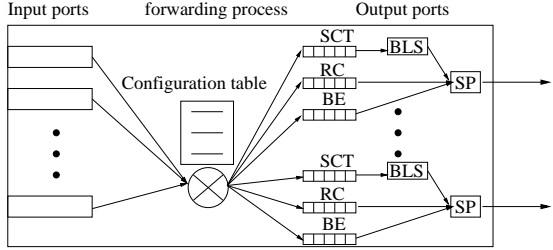}
		\footnotesize \caption{An Extended AFDX switch architecture}
		\label{fig:sw_archi}
	\end{figure}
	
	The legacy AFDX standard manages the exchanged data through the Virtual Link (VL) concept. This concept provides a way to reserve a guaranteed bandwidth for each traffic flow. The VL represents a multicast communication, which originates at a single End System and delivers packets to a fixed set of End Systems. Each VL is characterized by: (i) BAG (Bandwidth Allocation Gap), ranging in powers of 2 from 1 to 128 milliseconds, which represents the minimal inter-arrival time between two consecutive frames; (ii) MFS (Maximal Frame Size), ranging from 64 to 1518 bytes, which represents the size of the largest frame sent during each BAG. Furthermore, the legacy AFDX specifies a NP-SP scheduler based on two priority levels within End Systems and switches to enable the QoS features.
	
	In Fig.\ref{fig:sw_archi}, we illustrate the architecture of the extended AFDX switch. It consists of: (i) store and forward input ports to verify each frame correctness before sending it to the corresponding output port; (ii) a static configuration table to forward the received frames to the correct output port(s) based on their VL identifier; (iii) the output ports with three priority queues, multiplexed with a NP-SP scheduler, and the highest one is shaped with the BLS. 
	
	The current AFDX switch distinguishes the flow priority level based on its VL identifier stored in the static configuration table, i.e., for each VL identifier, there is a predefined priority level stored in the table. Hence, to manage both extra AFDX profiles, i.e., SCT and BE, we need to update the configuration table to add the corresponding VL identifiers and their associated priority levels.
	
	Hence, in comparison to the legacy AFDX switch architecture, the main modifications required for the proposed extended AFDX switch consists in: (i) at the software level, updating the static configuration table to manage three priority levels instead of two, Note that the update overhead is very limited since only one additional bit per line is necessary; (ii) at the hardware level, adding an extra priority queue at the output port since the legacy AFDX supports already two priorities, and implementing the BLS for the SCT queue on top of the NP-SP scheduler, as illustrated in Fig.\ref{fig:sw_archi}. 
	
	From the global avionic communication architecture point of view, the extended AFDX standard necessitates the update of the End-System at the application layer to enable a consistent mapping between VL identifiers and the appropriate priority level. 

	Moreover, the implementation and certification of this extended AFDX may imply extra costs, in comparison with the legacy one. However, this fact is counterbalanced by the major pros of such an homogeneous architecture, in terms of enhancing performance and reducing cables and weight.
	
	\subsection{BLS Shaper}
	\label{Spec}
	The BLS belongs to the credit-based shapers class. It has been defined in \cite{Gotz2012} by an upper threshold, $L_M$, a lower threshold $L_R$, such as $0\leq L_R < L_M$, and a reserved bandwidth, $BW$. Additionally, the priority of a class $j$ shaped by BLS, denoted $p(j)$, can oscillate between a high and a low value. The low value is usually below the lowest priority of unshaped traffic. In the avionic context, to guarantee the safety isolation level between the different traffic profiles, the low value associated to the SCT is set to be lower than the RC priority level, but higher than the BE's one. Therefore, SCT queue priority oscillates between 0 and 2, RC's priority is 1 and BE has the priority 3. Thus, when SCT traffic is enqueued, RC is the only traffic that can be sent and it happens when the SCT priority is 2. As a consequence, the BE traffic is isolated from the SCT and RC traffics.
	
	The credit counter varies as follows:\\
	(i) initially, the credit counter starts at 0 and the queue of the burst limited flows is high;\\
	(ii) the main feature of the BLS is the change of priority $p(j)$ of the shaped queue, which occurs in two contexts: 1) if $p(j)$ is high and  credit reaches $L_M$;  2) if $p(j)$ is low and credit reaches $L_R$; \\
	(iii) when a frame is transmitted, the credit increases (is consumed) with a rate of $I_{send}$, else the credit decreases (is gained) with a rate of $I_{idle}$;\\
	(iv) when the credit reaches $L_M$, it stays at this level until the end of the transmission of the current frame (if any);	\\
	(v) when the credit reaches $0$ it stays at this level until the end of the transmission of the current frame (if any). The credit remains at 0 until a new BLS frame is transmitted.\\
	\begin{figure}[h]
		\centering
		\includegraphics[width=0.40\textwidth]{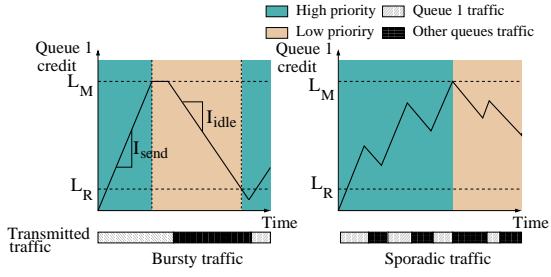}
		\footnotesize \caption{BLS (queue 1) credit evolution}
		\label{fig:BLScredit}
	\end{figure}	
	
	The behaviour of the BLS is illustrated in Fig. \ref{fig:BLScredit}, where the queue 1 is the shaped queue by the BLS. As shown, the credit is always between 0 and $L_M$. The different parameters of the BLS shaper are defined as follows:\\
	(i) the decreasing rate is: $$	I_{idle} = BW\cdot, C$$
	where $C$ is the link speed and $BW$ is the percentage of bandwidth reserved for BLS frames.\\
	(ii) the increasing rate is: $$	I_{send} = C - I_{idle}.$$
	

	Therefore, the implementation of BLS at the hardware level necessitates a counter to keep track of the credit and a timer to handle credit updates. These parameters, i.e., a counter and a timer, induce low extra complexity of implementing a BLS on top of a NP-SP scheduler, in comparison with a regular NP-SP scheduler. It is worth noting that with the BLS, both the priority of the shaped queue and the state of all the queues, i.e., empty or not, define whether the credit is gained or lost. This aspect is depicted in Fig. \ref{fig:BLScredit} for two arrival scenarios. The first one (left figure) shows the case of a bursty traffic, where the maximum of traffic shaped by the BLS is sent when its priority is the highest. Consequently, the other priorities send as much traffic as possible, when the BLS queue priority has the low value. The second one (right figure) is for sporadic traffic, where we can see that when the queue 1 priority is highest but no frame is available, then the credit is regained. However, when the priority is at the low value and the other queues are empty, then frames can be transmitted and the credit is consumed.

	\section{Timing Analysis Methodology}
	\label{TAM}
	We present in this section the worst-case timing analysis methodology based on Network Calculus (NC), and followed to conduct the performance evaluation of our proposed Extended AFDX. We first present the Network Calculus framework and define the necessary and sufficient schedulability conditions. Then, we detail the models of traffic flows, end-systems and switches, which will be extended in the next section to integrate the impact of the BLS on the different traffic classes. Finally, we explain the computation of the upper bounds on end-to-end delays. The main notations used in this paper are presented in Table \ref{notations2}, where upper indices indicate nodes or components and lower indices indicate traffic classes, flows or priority levels.
	

	\begin{table}[h!]
		\footnotesize 
		
		\begin{center}
			\begin{tabularx}{\linewidth} {l X}
				\hline
				$C$ & Link speed\\
				$MFS_k$ & Maximum Frame Size of flow $k$\\
				$BAG_k$ & Bandwidth Allocation Gap of flow $k$\\
				$J_k,Dl_k$ & Jitter and deadline of flow $k$\\				
				$L_M, L_R$ & BLS maximum and resume credit levels\\
				$I_{idle},I_{send}$ & BLS idle and sending slopes \\
				$BW$ & BLS reserved bandwidth\\			
				$p(j)$ & Priority level of a class $j$ with $p(j) \in \{0,1,2,3\} $ \\ 
				$UR^{bn}_{j}$& The bottleneck network utilisation rate of a class $j$ \\
				$n_{j}^{es}$ & Number of flows of class $j$ generated per node $es$\\
				$\beta_{j,k}^{n}$ & Service curve guaranteed to the flow k of class $j$ in a node $n \in \{es, sw\}$ or component $n \in \{bls, sp \}$ in its path\\
				$\beta_{j}^{n}$ & Service curve guaranteed for the traffic class $j$ in a node $n  \in \{es, sw \}$ or component $n \in \{bls, sp \}$\\
				$\beta_{SCT, n}^{sp}$ & Service curve guaranteed to the SCT when having the high priority level (n=0), or the low priority (n=2) in a $sp$ component \\
				$\alpha_{j,k}^{n,i}$ & Input arrival curve of the flow k of class $j$ in the $i^{th}$ node $n  \in \{es, sw \}$ or component $n  \in \{bls, sp \}$ in its path \\
				$\alpha_j^n$ & Input arrival curve of the aggregated flows of class $j$ in a node $n  \in \{es, sw\}$ or component $n  \in \{bls, sp \}$ \\
				$\alpha_{j,k}^{*,n,i}$ & Output arrival curve of the flow k of class $j$ from the $i^{th}$ node $n  \in \{es, sw\}$ or component $n  \in \{bls, sp \}$ in its path\\
				$\alpha_j^{*,n}$ & Output arrival curve of the aggregated flows of class $j$ from a node $n  \in \{es, sw \}$\ or a component $n  \in \{bls, sp \}$\\
				$\Delta_{i}^{j}$ &  $i\in \{send,idle\}$, and $j\in\{max,min\}$ the BLS windows defined in  Eq.(\ref{IidleSCMin}), Eq.(\ref{IidleSCMax}), Eq.(\ref{IsendSCMax}), and Eq.(\ref{IsendSCMax0}) \\
				\hline
			\end{tabularx}
		\end{center}
		\caption{Notations}
		\label{notations2}
	\end{table}

	\subsection{Network Calculus Framework}
	\label{Background}
	The timing analysis detailed in this paper is based on Network Calculus theory \cite{leboudecthiran12} providing upper bounds on delays and backlogs. Delay bounds depend on the traffic arrival described by the so called \textit{arrival curve} $\alpha$, and on the availability of the traversed node described by the so called minimum \textit{service curve} $\beta$. The definitions of these curves are explained as following.
	\begin{Definition}[Arrival Curve]
		\label{def:arrivalCurve}
		\cite{leboudecthiran12} A function $\alpha(t)$ is an arrival curve for a data flow with an input cumulative function $R(t)$,i.e., the number of bits received until time $t$, iff:
		\begin{displaymath}
		\forall t, R(t) \leq  R \otimes\footnote{$f \otimes g (t) = \inf_{0 \leq s \leq t}\{f(t-s) + g(s)\}$} \alpha(t)
		\end{displaymath}
	\end{Definition}
	\begin{Definition}[Strict minimum service curve]
		\label{def:strict-min-service-curve}
		\cite{leboudecthiran12} The function $\beta $ is the minimum \emph{strict} service curve for a data flow with an output cumulative function $R^*$, if for any backlogged period $]s,t]$\footnote{ $]s,t]$ is called backlogged period if $R(\tau) -R^*(\tau) >0, \forall \tau \in ]s,t] $}, $R^*(t) - R^*(s) \geq \beta(t-s)$.
	\end{Definition}
	\begin{Definition}[Maximum service curve]\label{def:max-service-curve}
		\cite{leboudecthiran12} The function $\gamma(t)$ is the maximum  service curve for a data flow with an input cumulative function $R(t)$ and output cumulative function $R^*(t)$ iff:
		\begin{displaymath}
		\forall t,	R^*(t) \leq R \otimes \gamma (t)
		\end{displaymath}
	\end{Definition}
	The traffic contracts are generally enforced using a leaky-bucket shaper, i.e., the traffic flow is $(r,b)$-constrained and the arrival curve is $\alpha (t)=  r\cdot t+b$ for $t >0$. A common model of service curve is the rate-latency curve $\beta_{R,T}$, defined as $ \beta_{R,T}(t)= [R(t-T)]^+ $, where $\left[x\right]^+$ is the maximum between $x$ and $0$. 

	Then, to compute the main performance metrics we need the following results. 
	\begin{Corollary}(Left-over service curve - SP Multiplex)\cite{bouillard2009service}
		\label{cor:residual-service-curve}
		Consider a system with the strict service $\beta$ and $m$ flows crossing it, $f_1$,$f_2$,..,$f_m$. The maximum packet length of $f_i$ is $l_{i,max}$ and $f_i$ is $\alpha_i$-constrained. The flows are scheduled by the non-preemptive static priority (NP-SP) policy, where priority of $f_i > $  priority of $f_j \Leftrightarrow i < j$. For each $i \in \{2,..,m \}$, the strict service curve of $f_i$ is given by\footnote{$g_{\uparrow}(t) = \max\{0,\sup_{0 \leq s \leq t} g(s)\}$}:
		$$ (\beta - \sum_{j <i} \alpha_j - \max_{k \geq i} l_{k,max})_{\uparrow}$$

	\end{Corollary}
	\begin{Theorem} [Performance Bounds]\cite{leboudecthiran12}
		\label{PerformanceBounds}
		Consider a flow $F$ constrained by an arrival curve $\alpha$ crossing a system $\mathcal{S}$ that offers a minimum service curve $\beta$ and a maximum service curve $\gamma$. The performance bounds obtained at any time $t$ are:\\
		Backlog\footnote{v: maximal vertical distance} : $ \forall~t:~q(t)\leq v(\alpha,\beta)$\\
		Delay\footnote{h: maximal horizontal distance}: 	$ \forall~t:~d(t)\leq h(\alpha,\beta)$\\
		Tight Output arrival curve\footnote{$f \oslash g(t) = \sup_{s \geq 0}\{f(t+s) - g(s)\}$ }: $\alpha^*(t) =(\gamma\otimes\alpha) \oslash\beta (t)$
	\end{Theorem}
	\begin{Theorem} (Concatenation-Pay Bursts Only Once)\cite{bouillard2009service}
		\label{ConcatenationOfNodes}
		Assume a flow crossing two servers with respective service curves  $\beta_1$ and $\beta_2$. The system composed of the concatenation of the two servers offers a minimum service curve $\beta_1 \otimes \beta_2$.
	\end{Theorem}
		Finally, to compute end-to-end delay bounds of individual traffic flows, we need the following corollary.
		
	\begin{Corollary}
		\label{thm:residual-service-curve}
		(Individual flow residual service curve - Blind Multiplex) \cite{bouillard2009service} Consider $m$ flows $f_1$,$f_2$,..,$f_m$ of class $j$ crossing a system $n$ with the strict service $\beta_ {j}^{n}$ offered to class $j$ and with $f_k$ $\alpha_k^{sp}$-constrained. The maximum packet length of $f_k$ is $l_{k,max}$. Then the residual service curve offered to $f_k$ in a node $n$ in its path is:
		\begin{displaymath}
		\beta_{j,k}^{n}(t)= (\beta_ {j}^{n}(t) -  \sum_{l\neq k} \alpha_l^{sp}(t)-\max_{l \neq k} l_{l,max})_\uparrow
		\end{displaymath}
	
	\end{Corollary}
	
	\subsection{Schedulability Conditions}
	\label{scheduTest}
	To infer the real-time guarantees of our proposed solution, we need first to define a necessary schedulability condition. The latter consists in respecting the stability condition within the network, where the maximum arrival rate of an input traffic at any crossed node has to be lower than the minimum guaranteed service rate within the node. This constraint is denoted as \textbf{\textit{rate constraint}}. Then, we define a sufficient schedulability condition to infer the traffic schedulability, which consists in comparing the upper bound on end-to-end delay of each traffic flow to its deadline. This constraint is called \textbf{\textit{deadline constraint}}.
	
	For this sufficient schedulability condition, we detail the end-to-end delay expression of a flow $k$ in the class $j\in \{SCT,RC,BE\}$, $EED_{j,k}$, along its path $path_k$ as follows:
	\begin{equation}\label{Dend-2-end}
	EED_{j,k}= d_{j,k}^{es}+ d_{prop}+\sum\limits_{sw_i\in path_k}d_{j,k}^{sw_i}
	\end{equation}
	
	With $d_{j,k}^{es}$ the delay within the end-system $es$ to transmit the flow $k$ of class $j$ and $d_{prop}$ the propagation delay along the path, which is generally negligible in an avionics network. The last delay $d_{j,k}^{sw_i}$ represents the upper bound on the delay within each switch $sw_i$ along the flow path, and it consists of several parts: (i) the store and forward delay at the input port, equal to $\frac{L}{C}$, with $L$ the length of the frame and $C$ the capacity; (ii) the technological latency due to the switching process, upper-bounded by 1$\mu$s; (iii) the output port delay due to the BLS and NP-SP scheduler. Hence, the only two unknown are the delays in the end-system and the output port of the switch. To enable the computation of upper bounds on these delays, we need to model the different parts of the network, and more particularly the BLS.
	
	\subsection{System Modeling}
	\label{model}
	To compute upper bounds on end-to-end delays of different traffic classes using Network Calculus, we need to model each message flow to compute its maximum arrival curve, and the behaviour of end-systems and the crossed switches to compute the minimum service curves.
	
	The characteristics of each traffic flow $k$ of class $j\in \{SCT,RC,BE\}$, generated by an end-system, is characterised by $\left( BAG_{k}, MFS_{k}, Dl_{k}, J_{k}\right) $ for respectively the minimum inter-arrival time, the maximum frame size integrating the protocol overhead, the deadline $Dl_{k}$ if any (generally equal to $BAG_{k}$ unless explicitly specified and infinite for BE) and the jitter.
	
	The arrival curve of traffic class $j$ in the end-system $es$, based on a leaky bucket model, is as follows:	
	$$\alpha_{j}^{es}(t) = \sum\limits_{k \in j}\alpha_{j,k}^{es}(t)= \sum\limits_{k \in j}MFS_{k}+\frac{MFS_{k}}{BAG_{k}}\left( t+J_{k}\right)$$	
	\[ 
	\alpha_{j}^{es}(t) = b_j+r_jt \text{ with }
	\begin{cases}
	b_j=\sum\limits_{k \in j}MFS_{k}+\frac{MFS_{k} }{BAG_{k}}J_{k}\\
	r_j=\sum\limits_{k \in j}\frac{MFS_{k} }{BAG_{k}}
	\end{cases}
	\]	
	For the end-systems, they are implementing a Non-Preemptive Static Priority Scheduler (NP-SP). This scheduler has been already modelled in the literature \cite{bouillard2009service} through Cor. \ref{cor:residual-service-curve}, and the defined strict minimum service curve guaranteed to a traffic class $j\in \{SCT,RC,BE\}$ within an end-system $es$ is as follows:\\	
	$\beta_{j}^{es}(t) =\left[ C\cdot t- \sum\limits_{ k \in i, p(i) < p(j)}\alpha_{i,k}^{es}(t) - \max\limits_{k \in i, p(i) \geq p(j)} MFS_k \right]_ \uparrow $
	
	For the proposed Extended AFDX switches, we need first to model the impact on the SCT of the BLS implemented on top of the NP-SP scheduler. To achieve this aim, we distinguish two possible scenarios. The first one covers the particular case where the priority of SCT remained low (2), i.e., the other queues are empty; whereas the second one covers the general case where the priority of SCT oscillates between low (2) and high (0), as explained in Section \ref{Spec}. The minimum service curve guaranteed within the switch in the first scenario is due to the NP-SP scheduler and denoted $\beta^{sp}_{SCT, 2}$, which is computed via Cor. \ref{cor:residual-service-curve} when considering the impact of RC traffic as the highest priority and the BE as the lowest priority. On the other hand, the minimum service curve guaranteed within the switch in the second scenario is computed via Th. \ref{ConcatenationOfNodes}, through the concatenation of the service curves within the BLS $\beta_{SCT}^{bls}$ (computed in Section \ref{min-SCT}) and the NP-SP $\beta_{SCT, 0}^{sp}$ (computed via Cor. \ref{cor:residual-service-curve}). Therefore, we define the following relations between the service curves guaranteed within the switch $sw$ and the components $\{bls, sp\}$ for the traffic classes $\{SCT, RC \} $, $\beta_{SCT}^{sw}$ and $\beta_{RC}^{sw}$, respectively:
	\begin{eqnarray}
	&\beta_{SCT}^{sw} (t)=\max(\beta_{SCT, 2}^{sp},\beta_{SCT, 0}^{sp}\otimes\beta_{SCT}^{bls} (t))\label{eq:swSCT}\\
	&\beta_{RC}^{sw}(t) =\left[ C\cdot t- \alpha_{SCT}^{*,bls}(t) - \max\limits_{k \in i, p(i) \geq p(RC) } MFS_k \right ]_ \uparrow\nonumber
	\end{eqnarray}
	
	In Section \ref{SC-NC}, we will detail the minimum service curve guaranteed within the BLS $\beta_{SCT}^{bls}$ and the maximal output arrival rate $\alpha^{*,bls}_{SCT}$ depending on the respective maximum service curve $\gamma_{SCT}^{bls}$.

	\subsection{Computing End-to-End Delays}
	The computation of the end-to-end delay upper bounds follows four main steps:
	
	(1) Computing the strict minimum service curve guaranteed to each traffic class $j\in \{SCT,RC,BE\}$ in each node $n$ of type $n\in\{es,sw\}$, $\beta_{j}^{n}$, will infer the computation of the residual service curve, guaranteed to each individual flow $k$ of class $j$, $\beta_{j,k}^{n}$ using Cor. \ref{thm:residual-service-curve};
	                                                                                                                                                                                                                                                                                                    
	(2) Knowing the residual service curve guaranteed to each flow within each crossed node allows the propagation of the arrival curves along the flow path, using Th.\ref{PerformanceBounds}. We can compute the output arrival curve, based on the input arrival curve and the minimum service curve, which will be in its turn the input of the next node;
	
	(3) The computation of the minimum end-to-end service curve of each flow $k$ in class $j$, based on Th.\ref{ConcatenationOfNodes}, is simply the concatenation of the residual service curves, $\beta_{j,k}^{n}$, $\forall n$ along its path $path_k$.
	
	(4) Given the minimum end-to-end service curve of each flow $k$ in class $j$ along its $path_k$ and its maximum arrival curve at the initial source, the end-to-end delay upper bound $D_{j,k}$ is the maximum horizontal distance between the two curves (see Th.\ref{PerformanceBounds}).
	
	Hence, as we can notice, we need to model all the unknown service curves, related to the BLS, to enable the end-to-end delay upper bounds computation. These curves are detailed in the next section. It is worth noting that since the BE class has no deadline, the service curves guaranteed to this class and the computation of the respective upper bounds on end-to-end delays are not detailed in this paper.

	\section{Service Curves using NC}
	\label{SC-NC}
	To compute the service curves offered to the burst limited traffic flows, i.e., SCT, and the non burst limited traffic flows, i.e., RC and BE, we need to detailed two types of windows, which are enforced by the BLS behaviour.  The first one is denoted as \textit{sending window}, during which the SCT has the highest priority and is sent until the consumed credit reaches the maximum threshold, $L_M$. The second one is called \textit{idle window} where the SCT has the priority just higher than BE and the consumed credit is decreasing until reaching the minimum threshold, $L_R$. Moreover, due to the non-preemptive message transmission, both windows have minimal and maximal durations. 
	The various combinations of such durations will induce the service curves, which are necessary for computing upper bounds on end-to-end delays and detailed in this section. 
	
	\subsection{Strict Minimum Service Curve of SCT}
	\label{min-SCT}
	The strict minimum service curve of SCT defines a lower bound on the SCT output cumulative traffic from the BLS. This curve represents the most deteriorated behaviour of BLS, in terms of offered service to the SCT, which maximizes its delay within the BLS. Hence, to cover this worst-case behaviour, we combine the maximum \textit{idle window}  and the minimum \textit{sending window} durations.
	
	The minimum \textit{sending window} duration, $\Delta_{send}^{min}$, is the time for the consumed credit to go from the lowest to the highest thresholds, i.e., from $L_R$ to $L_M$, with an increasing slope $I_{send}$: 
	\begin{equation}
	\label{IsendSCMin}
	\Delta_{send}^{min} = \frac{L_M-L_R}{I_{send}}
	\end{equation}
	
	The maximum \textit{idle window} duration, $\Delta_{idle}^{max} $, is the time for the consumed credit to go from $L_M$ to $L_R$ with a decreasing slope $I_{idle}$, in addition to the transmission time of a maximum frame of the RC traffic. The latter is due to the non-preemption feature when a RC frame is starting its transmission just before the consumed credit reaches the lowest threshold, $L_R$. It is worth noting that the BE class impacts the SCT only within the NP-SP scheduler and not within the BLS since it has a priority (3) lower than the lowest priority of SCT (2):
	\begin{equation}
	\label{IidleSCMax}
	\Delta_{idle}^{max} = \frac{L_M-L_R}{I_{idle}} +\frac{MFS_{RC}}{C}
	\end{equation}
	
	Therefore, the strict minimum service curve guaranteed to the SCT, $\beta_{SCT}^{bls}$ is defined in Th. \ref{Th:SCT}.
	\begin{Theorem}[Strict Minimum Service Curve of SCT in BLS] 
		\label{Th:SCT}
		Consider a SCT crossing a server with a constant rate $C$, implementing a BLS shaper. The strict minimum service curve guaranteed to the SCT is as follows:
		\begin{equation}
		\label{betaRateLatencySCT}
		\beta_{SCT}^{bls}(t) =\frac{\Delta^{min}_{send}}{\Delta_{send}^{min}+\Delta_{idle}^{max}}  \cdot C \cdot \left( t-\Delta_{idle}^{max}\right)^+ 
		\end{equation}
		where $\left[x\right]^+$ is the maximum between $x$ and $0$.
	\end{Theorem}
	
	\begin{proof} 
		\label{proof1}
		Consider $R^*_{SCT} (t)$ the output cumulative function of the SCT at the output of the server implementing a BLS, and $\Delta R^*_{SCT}(\delta)$ the variation of the output cumulative function during $\delta$. To prove that $\beta_{SCT}^{bls}$ in Eq. (\ref{betaRateLatencySCT}) is a strict minimum service curve, we need to prove Def. \ref{def:strict-min-service-curve} for any backlogged period $\delta$, i.e., the SCT flows are continuously backlogged during $\delta$. 
		
		During a backlogged period $\delta$, the SCT has at least $p$ opportunities of full service constrained by $\beta (t) = C\cdot t$ during the minimum sending window $\Delta_{send}^{min}$, then:
		\begin{equation}
		\Delta R^*_{SCT} (\delta)\geq p\cdot C\cdot\Delta_{send}^{min}
		\label{min-src-sct}
		\end{equation}
		
		The main idea is to find a lower bound of $p$ to define the service curve guaranteed to SCT, $\beta_{SCT}^{bls}$. On the other hand, if SCT has $p$ opportunities to be transmitted, then the RC traffic (since it is the only traffic class with a priority higher than the lowest priority of SCT during the idle window) has at most $(p+1)$ opportunities to be transmitted during at the worst-case the maximum idle window, $\Delta_{idle}^{max}$, then:
		\begin{equation}
		\Delta R^*_{RC} (\delta) \leq (p+1) \cdot C\cdot\Delta_{idle}^{max}
		\label{max-src-rc}
		\end{equation}
		
		Giving the strict service curve property of $C\cdot t$ since we have a constant rate server and using Eq. (\ref{max-src-rc}), we have: 
		\begin{eqnarray}
		C\cdot\delta &\leq & \Delta R^*_{SCT}(\delta) +\Delta R^*_{RC}(\delta) \nonumber \\
		&  \leq & \Delta R^*_{SCT}(\delta)+(p+1)\cdot C\cdot\Delta_{idle}^{max}\nonumber
		\end{eqnarray}
		
		Consequently, the lower bound of $p$ is as follows:
		\begin{eqnarray}
		\label{LBP}
		p&\geq & \frac{C\cdot\delta-\Delta R^*_{SCT}(\delta)}{C\cdot\Delta_{idle}^{max}}-1
		\end{eqnarray}
		
		When injecting Eq.(\ref{LBP}) in Eq. (\ref{min-src-sct}), we obtain:
		\begin{eqnarray}
		&& \Delta R^*_{SCT}(\delta) \geq  \left(\frac{C\cdot\delta-\Delta R^*_{SCT}(\delta)}{C\cdot\Delta_{idle}^{max}}-1 \right)\cdot C\cdot\Delta_{send}^{min}\nonumber\\
		&& \Delta R^*_{SCT}(\delta)\cdot\left( 1+ \frac{\Delta_{send}^{min}}{\Delta_{idle}^{max}}\right) \geq  \left(\frac{ \delta}{\Delta_{idle}^{max}}-1\right)\cdot C\cdot	\Delta_{send}^{min}	\nonumber\\
		&& \Delta R^*_{SCT}(\delta) \geq  \frac{\frac{ \delta}{\Delta_{idle}^{max}}-1}{ \frac{\Delta_{send}^{min}}{\Delta_{idle}^{max}}+1}\cdot	C\cdot\Delta_{send}^{min}\nonumber
		\end{eqnarray}
		
		Giving that $\Delta R^*_{SCT}(\delta) \geq 0$, then:
		\[ \Delta R^*_{SCT}(\delta) \geq  \frac{\Delta^{min}_{send}}{\Delta_{send}^{min}+\Delta_{idle}^{max}}  \cdot C \cdot \left( \delta -\Delta_{idle}^{max}\right)^+  \]	
	\end{proof}

	\subsection{Maximum Service Curve of SCT}
	The maximum service curve of SCT represents the best offered service to the SCT, which induces the minimum processing delay within the BLS. As such, in the presence of RC traffic, we combine the minimum \textit{idle window} duration and the maximum \textit{sending window} one to handle this best-case behaviour. 
	
	The maximum \textit{sending window} duration, $\Delta_{send}^{max}$, is equal to the sum of : (i) the minimum \textit{sending window} duration, $\Delta_{send}^{min}$; (ii) the transmission time of a maximum frame of the SCT due to the non-preemption feature, i.e., one SCT frame may start its transmission just before the consumed credit reaches $L_M$; (iii) the time to consume the gained credit during the transmission of one additional maximum frame of RC traffic at the end of the \textit{idle window}. The latter parameter is due to the fact that the resume level of BLS, $L_R$, is the lower threshold on the consumed credit to trigger the priority change of the SCT from lowest to highest, and not an extreme value for the consumed credit itself. Actually, if a frame of RC traffic has been transmitted just at the end of the \textit{idle window}, 
	the consumed credit keeps decreasing until it either reaches 0, or the transmission ends. Therefore, the lowest value the consumed credit can reach due to the non-preemption feature is $\max(0,L_R-\frac{MFS_{RC}}{C}.I_{idle})$. The additional time during which the consumed credit can then increase with a slop $I_{send}$ is $ \frac{L_R - \max(0,L_R-\frac{MFS_{RC}}{C}.I_{idle})}{I_{send}}$. 
	The maximum \textit{sending window} duration is then as follows:		
	\begin{equation}\label{IsendSCMax}
	\Delta_{send}^{max} = \frac{L_M-L_R}{I_{send}} + \frac{MFS_{SCT}}{C}+\min(\frac{MFS_{RC}}{C}\cdot \frac{I_{idle}}{I_{send}},\frac{L_R}{I_{send}})
	\end{equation}
	However, it is worth noting that the consumed credit may start at $0$, such as at the initialisation phase or after a long period of inactivity. Hence, the maximum \textit{sending window} duration covering such possibility, $\Delta_{send,0}^{max}$, is as follows:
	\begin{equation}
	\label{IsendSCMax0}
	\Delta_{send,0}^{max} = \frac{L_M}{I_{send}} +\frac{MFS_{SCT}}{C} 
	\end{equation}
	
	The minimum \textit{idle window} duration, $\Delta_{idle}^{min} $, is simply the time it takes for the consumed credit to go from $L_M$ to $L_R$ with a decreasing slope of $I_{idle}$:
	\begin{equation}
	\label{IidleSCMin}
	\Delta_{idle}^{min} = \frac{L_M-L_R}{I_{idle}}
	\end{equation}

	Therefore, the maximum service curve guaranteed to the SCT, $\gamma_{SCT}^{bls} $ is defined in Th. \ref{Th:SCT-max}.
	\begin{Theorem}[Maximum Service Curve of SCT in BLS] 
		\label{Th:SCT-max}
		Consider a SCT crossing a server with a constant rate $C$, implementing a BLS shaper. The maximum service curve guaranteed to the SCT is as follows. 
		\begin{equation}
		\label{gammaRateLatency}
		\gamma_{SCT}^{bls} (t) = \left\{ \begin{array}{ll}
		\textrm{ if no RC traffic: } C \cdot t \nonumber\\
		\textrm{Otherwise:}\nonumber\\
		\frac{\Delta_{send}^{max}}{\Delta^{nom}_{\gamma_{SCT}}}\cdot C\cdot t +\Delta_{send,0}^{max}\cdot C \cdot\frac{\Delta_{idle}^{min}}{ \Delta^{nom}_{\gamma_{SCT}}}\nonumber
		\end{array} \right.
		\end{equation}			
		with $\Delta^{nom}_{\gamma_{SCT}}=\Delta_{send}^{max}+\Delta_{idle}^{min}$.
	\end{Theorem}
	
	\begin{proof} 
		First, it is obvious that in the absence of RC traffic, SCT can use the maximum service $\gamma(t) = C.t$. Then, for the more general case, consider $R^*_{SCT} (t)$ the output cumulative function of the SCT at the output of the server implementing a BLS, and $\Delta R^*_{SCT}(\delta)$ the variation of the output cumulative function during $\delta$. 
		
		During a backlogged period $\delta$ for SCT and RC traffic, the SCT has at most $p+1$ opportunities with $p$ times a full service constrained by $\gamma$ during the maximum sending window $\Delta_{send}^{max}$, in addition to once during $\Delta_{send,0}^{max}$, then:
		\begin{equation}
		\Delta R^*_{SCT} (\delta) \leq p\cdot C\cdot \Delta_{send}^{max}+C\cdot\Delta_{send,0}^{max}
		\label{max-src-sct}
		\end{equation}
		The main idea is to find an upper bound of $p$ to define the maximum service curve guaranteed to SCT, $\gamma_{SCT}^{bls}$. On the other hand, if SCT has at most $p+1$ opportunities to be transmitted, then the RC traffic has at least $p$ opportunities to be transmitted during the minimum idle window, $\Delta_{idle}^{min}$, then:
		\begin{equation}
		\Delta R^*_{RC} (\delta) \geq p \cdot C\cdot\Delta_{idle}^{min}
		\label{min-src-rc}
		\end{equation}
		Giving the maximum service curve property of $\gamma$ and using Eq. (\ref{min-src-rc}), we have: 
		\begin{eqnarray}
		C\cdot\delta &\geq & \Delta R^*_{SCT}(\delta) +\Delta R^*_{RC}(\delta) \nonumber \\
		&  \geq & \Delta R^*_{SCT}(\delta)+p\cdot C\cdot\Delta_{idle}^{min}\nonumber
		\end{eqnarray}
		Consequently, the upper bound of $p$ is as follows:
		\begin{eqnarray}
		\label{UBP}
		p &\leq & \frac{C\cdot\delta-\Delta R^*_{SCT}(\delta)}{C\cdot\Delta_{idle}^{min}}
		\end{eqnarray}
		When injecting Eq.(\ref{UBP}) in Eq.(\ref{max-src-sct}), we obtain:
		
		$$\Delta R^*_{SCT}(\delta) \leq   \frac{C\cdot\delta-\Delta R^*_{SCT}(\delta)}{C\cdot\Delta_{idle}^{min}}\cdot C\cdot\Delta_{send}^{max}+C\cdot\Delta_{send,0}^{max}$$
		$$ \Delta R^*_{SCT}(\delta)\cdot\left( 1+\frac{\Delta_{send}^{max}}{\Delta_{idle}^{min}}\right)  \leq   \frac{\delta}{\Delta_{idle}^{min}}\cdot C\cdot\Delta_{send}^{max}+C\cdot\Delta_{send,0}^{max}$$
		$$\Delta R^*_{SCT}(\delta) \leq   \frac{\frac{\delta}{\Delta_{idle}^{min}}\cdot C\cdot\Delta_{send}^{max}+C\cdot\Delta_{send,0}^{max}}{ 1+\frac{\Delta_{send}^{max}}{\Delta_{idle}^{min}}} $$
		$$ \Delta R^*_{SCT}(\delta) \leq \frac{\Delta_{send}^{max}}{\Delta^{nom}_{\gamma_{SCT}}}\cdot C\cdot \delta +\Delta_{send,0}^{max}\cdot C \cdot\frac{\Delta_{idle}^{min}}{ \Delta^{nom}_{\gamma_{SCT}}}$$
		where $\Delta^{nom}_{\gamma_{SCT}}=\Delta_{send}^{max}+\Delta_{idle}^{min}$.
	\end{proof}
	
	\subsection{Maximum Output Arrival Curve of SCT}
	
	The maximum output arrival curve of SCT, from the BLS is detailed in the following Corollary:
	\begin{Corollary} [Maximum Output Arrival Curve of SCT from BLS]
		\label{cor:SCT-max-arrival}
		Consider a SCT with a maximum leaky-bucket arrival curve $\alpha$ at the input of a BLS shaper, guaranteeing a minimum rate-latency service curve $\beta_{SCT}^{bls}$ and a maximum service curve $\gamma_{SCT}^{bls}$. The maximum output arrival curve is:
		\begin{equation}
		\label{MaxArrival}
		\alpha^{*,bls}_{SCT} (t)=\min(\gamma_{SCT}^{bls}(t),\alpha \oslash \beta_{SCT}^{bls}(t) )
		\end{equation}
	\end{Corollary}
	
	\begin{proof}
		To prove Corollary \ref{cor:SCT-max-arrival}, we generalize herein the rule 13 in p. 123 in \cite{leboudecthiran12}, i.e., $(f\otimes g)\oslash g \leq f \otimes (g\oslash g) $, to the case of three functions $f$, $g$ and $h$ when $g\oslash h \in \mathcal{F}$, where $\mathcal{F}$ is the set of non negative and wide sense increasing functions:
		\begin{displaymath}
		\mathcal{F}=\{ f: \mathbb{R}^+ \rightarrow \mathbb{R}^+ \mid f(0)=0, \forall t \geq s: f(t) \geq f(s) \}
		\end{displaymath}
		According to Theorem \ref{PerformanceBounds}, we have $\alpha^*(t)= (\gamma_{SCT}^{bls}\otimes\alpha)\oslash \beta_{SCT}^{bls}$. Moreover, in the particular case of a leaky-bucket arrival curve $\alpha$ and a rate-latency service curve $\beta_{SCT}^{bls}$, $\alpha\oslash \beta_{SCT}^{bls}$ is a leaky-bucket curve, which is in $\mathcal{F}$. Hence, we have the necessary condition to prove the following:		
		$$(\alpha\otimes\gamma)\oslash\beta(t)\leq  \gamma\otimes(\alpha\oslash\beta)(t)  \leq \min(\gamma(t),\alpha\oslash\beta(t))   $$
	
	\end{proof}	
	\subsection{Minimum Strict Service Curves of RC}
	\begin{Theorem} [Minimum Strict Service Curves of RC]
		\label{th:bls-min-servicel}
		Consider a SCT with a maximum leaky-bucket arrival curve $\alpha$ at the input of a server with a constant rate $C$ implementing a BLS shaper, guaranteeing a minimum rate-latency service curve $\beta_{SCT}^{bls}$ and a maximum service curve $\gamma_{SCT}^{bls}$. The server is crossed by SCT and RC traffic, where SCT has the highest priority. The minimum strict service curve guaranteed to RC traffic in the NP-SP, integrating the impact of the BLS, is as follows:		
		\begin{eqnarray}
		\label{betaNSCTmax}
		\beta_{RC}^{sw} (t) &=& \left[ \max(\beta_{RC}^{sp}(t), \beta_{RC}^{bls}(t)) - \max\limits_{k \in i, p(i) \geq p(RC) } MFS_k \right ]_ \uparrow \nonumber
		\end{eqnarray}		
		\begin{eqnarray}
		\nonumber
		\textrm{where: } \beta_{RC}^{sp} (t) & = &(C \cdot t-\alpha_{SCT}^{bls}\oslash\beta_{SCT}^{bls}(t))^+\\
		\beta_{RC}^{bls} (t)&= & \frac{\Delta_{idle}^{min}}{\Delta^{nom}_{\gamma_{SCT}}}\cdot C\cdot \left( t-\Delta_{send,0}^{max}\right)^+ \nonumber\\
		\textrm{with: }\Delta^{nom}_{\gamma_{SCT}}&=&\Delta_{send}^{max}+\Delta_{idle}^{min}
		\end{eqnarray}	
	\end{Theorem}
	
	\begin{proof}
		According to Cor. \ref{cor:residual-service-curve}, the residual minimum strict service curve guaranteed to RC traffic crossing a NP-SP scheduler is as follows:
		\begin{equation}
		\label{Eq:RC}
		\beta_{RC}^{sw}(t) =\left[ C \cdot t- \alpha_{SCT}^{*,bls}(t) -  \max\limits_{k \in i, p(i) \geq p(RC) } MFS_k \right ]_\uparrow
		\end{equation}
		
		Moreover, according to Cor. \ref{cor:SCT-max-arrival}, the maximum output arrival curve of SCT from the BLS, $ \alpha_{SCT}^{*,bls}(t)$, is:
		\begin{equation}
		\label{eq:SCT}
		\alpha_{SCT}^{*,bls}(t)  =\min(\gamma_{SCT}^{bls}(t),\alpha \oslash \beta_{SCT}^{bls}(t))
		\end{equation}
		
		Using Eq. (\ref{Eq:RC}) and Eq. (\ref{eq:SCT}), we can deduce the following:
		\begin{eqnarray}
		\label{eq:rsc}
		&&\beta_{RC}^{sw} (t) = [ C \cdot t - \min(\alpha_{SCT}^{bls}\oslash\beta_{SCT}^{bls}(t), \gamma_{SCT}^{bls}(t)) \nonumber \\
		&&-\max\limits_{k \in i, p(i) \geq p(RC) } MFS_k  ]_ \uparrow \nonumber \\
		&& = [ \max ((C \cdot t - \alpha_{SCT}^{bls}\oslash\beta_{SCT}^{bls}(t))^+, (C \cdot t -\gamma_{SCT}^{bls}(t))^+ ) \nonumber \\
		&& - \max\limits_{k \in i, p(i) \geq p(RC) } MFS_k  ]_\uparrow \nonumber \\ 
		&& = [ \max ((C \cdot t - \alpha_{SCT}^{bls}\oslash\beta_{SCT}^{bls}(t))^+,\nonumber \\&& \frac{\Delta_{idle}^{min}}{\Delta^{nom}_{\gamma_{SCT}}}\cdot C\cdot \left( t-\Delta_{send,0}^{max}\right)^+ ) \nonumber \\
		&& - \max\limits_{k \in i, p(i) \geq p(RC) } MFS_k  ]_\uparrow \nonumber 
		\end{eqnarray}
		
	\end{proof}

	\section{Performance Analysis}
	\label{PA}
	In this section, we conduct performance analysis of the proposed Extended AFDX to evaluate its efficiency to support mixed criticality data, in comparison to the legacy AFDX solution (i.e. AFDX with regular NP-SP scheduler). This evaluation is based on the worst-case timing analysis methodology and the various service curves detailed in Sections \ref{TAM} and \ref{SC-NC}, respectively.
	First, we describe our realistic avionics case study. Afterwards, we assess the scalability of the Extended AFDX in comparison to the legacy one based on the current avionics traffic configuration. Finally, we analyse the impact of our proposal on SCT and RC performance, in comparison to the legacy solution, when considering future avionics traffic configurations.
	
	\subsection{Avionics Case Study}
	Our case study is a representative avionics communication architecture of the A380, based on a 1-Gigabit AFDX\footnote{The 1-Gigabit version of the AFDX is under specification.} backbone network, which consists of 4 switches and 64 end-systems as shown in Fig. \ref{fig:rlfullnetwork} (a). The maximum utilisation rate on the 100Mbps legacy AFDX on board the A380 is 30\%. Thus, on the 1Gigabit version, the maximum utilisation rate will be only of 3 \%. However, there is only legacy AFDX traffic, i.e., RC, circulating on this current communication architecture. Hence, to enable the performance analysis of our proposed Extended AFDX, we have extended this current traffic configuration, denoted herein as \textit{legacy reference configuration}, to support different traffic profiles generated by each end-system, which are described in Tab. \ref{table:multitrafficprofiles}. Each traffic class is characterised by $(MFS, BAG,$ $Dl, J)$ as detailed in Section \ref{model}. The SCT has a deadline of 2ms, and because of the BLS behaviour it admits the highest priority 0 during the sending windows and the priority below RC (2) during the idle windows. Each circulating traffic flow on the backbone network is a multicast flow with $16$ destinations, and crosses two successive switches before reaching its final destinations. The first switch in the path receives traffic from $16$ end-systems to forward it in a multicast way to its two neighbouring switches. Afterwards, the second switch in the path, which receives traffic from the two predecessor switches, forwards the traffic in its turn to the final end-system. Figure \ref{fig:rlfullnetwork} (b) shows the traffic communication patterns between the source and the final destinations of a given flow.
	\begin{table}[h!]
		\footnotesize
		\centering
		\begin{tabular}{|c|c|c|c|c|c|}
			\hline
			Priority & Traffic type & MFS & BAG & deadline& jitter  \\
			&   & (Bytes)          & (ms)&(ms) &(ms)\\
			\hline
			0/2 & SCT & 64 & 2 & 2 & 0 \\  
			\hline
			1 &  RC & 320 & 2 & 2 & 0\\
			\hline
			3 &	BE & 1024 & 8 & none & 0.5 \\ 
			\hline
		\end{tabular}
		\footnotesize \caption{Avionics flow Characteristics}
		\label{table:multitrafficprofiles}
	\end{table}
	
	The main considered performance metrics are: (i) the maximum utilisation rate of each traffic class, that can be sent on the Extended AFDX architecture while respecting the system constraints described in Section \ref{scheduTest}. This metric enables the scalability analysis of our proposal, in comparison with the legacy AFDX; (ii) the delay bound of SCT and RC classes to prove the predictability of the Extended AFDX and analyse its impact on the system timing performance, in comparison with the legacy AFDX. It is worth noting that because the BE does not have a deadline, and its largest impact on the other priorities is the transmission time of a maximum sized frame, the timing performance of this class is not detailed in this paper.
	
	To compute both performance metrics, we consider four scenarios, i.e., scenarios 1 and 2 are for the scalability analysis; whereas scenarios 3 and 4 are for timing analysis. 
	
	Concerning scenarios 1 and 2, we have started from the \textit{legacy reference configuration}, and then computed the maximum utilisation rate of SCT class that can be transmitted on the legacy AFDX, while respecting the system constraints in the presence of 3\% of RC traffic. This computation has shown that the legacy AFDX can support up to a maximum utilisation rate for SCT of $UR^{bn}_{SCT}=28.7\%$. Hence, scenario 1 (resp. scenario 2) consists in starting from a traffic configuration characterised by ($UR^{bn}_{SCT}=28.7\%; UR^{bn}_{RC}=3\%$) circulating on the Extended AFDX, then increasing the $UR^{bn}_{SCT}$ (resp. $UR^{bn}_{RC}$) until finding the maximum value which still respects the system constraints. The aim of this scenario is to compare the scalability of the Extended and legacy AFDX solutions, when increasing the congestion due to SCT (resp. RC) traffic. The results of both scenarios (1 and 2)  are detailed in Section \ref{scalability}.
	
	Afterwards, to have an idea about the timing performance of future avionics configurations based on the 1 Gigabit AFDX technology, which may very probably support higher utilisation rate of SCT and RC traffic than the \textit{legacy reference configuration}, we consider scenarios 3 and 4 described in Table \ref{table34}. As it can be noticed, the principle of scenario 3 (resp. scenario 4) is to fix the utilisation rate of RC class (resp. SCT class)  at 20\% and vary the SCT (resp. RC) utilisation rate to assess the impact of increasing network congestion on the timing performance. The considered BLS parameters are the same for both scenarios: $BW=0.46$ to support a maximum utilisation rate of SCT of $UR^{bn}_{SCT}=45\%$ (this is an upper bound for the estimated future needs in terms of SCT traffic); $L_R=0$ and $L_M=22,118$ to enable the transmission of a maximum SCT burst within the BLS of 80 frames during a minimum sending window, i.e., a generated burst of 5 SCT flows per End-System. Moreover, as it is illustrated in Table \ref{table34}, the variation of the utilisation rate of a class $i$ is obtained through increasing the number of generated traffic flows within each end-system, $n_{i}^{es}$. Thus, the bottleneck  utilisation rate is equal to $UR^{bn}_{i} = \frac{C_i}{C}$ with  $C_i$ the capacity used in the bottleneck by the aggregated flows of class $i\in \{RC,SCT\}$ and $C_i = 16 \cdot n_{i}^{es}\cdot \frac{MFS_{i}}{BAG_{i}}$. The results of both scenarios (3 and 4) are detailed in Section \ref{timing}.

	\begin{figure}[htbp]
		\centering
		\subfigure[]{\includegraphics[width=0.35\linewidth]{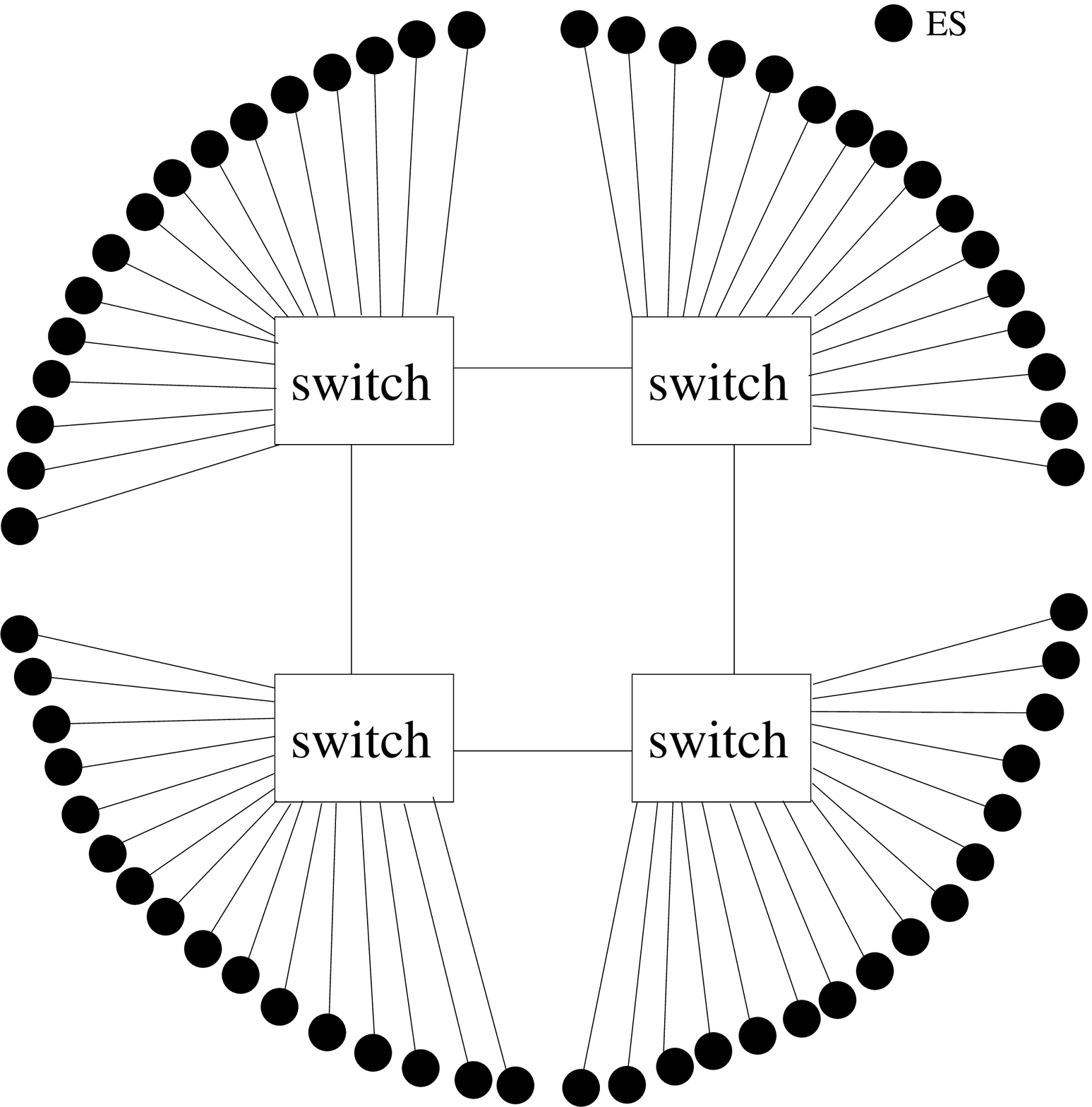}}
		\centering
		\subfigure[]{\includegraphics[width=0.45\linewidth]{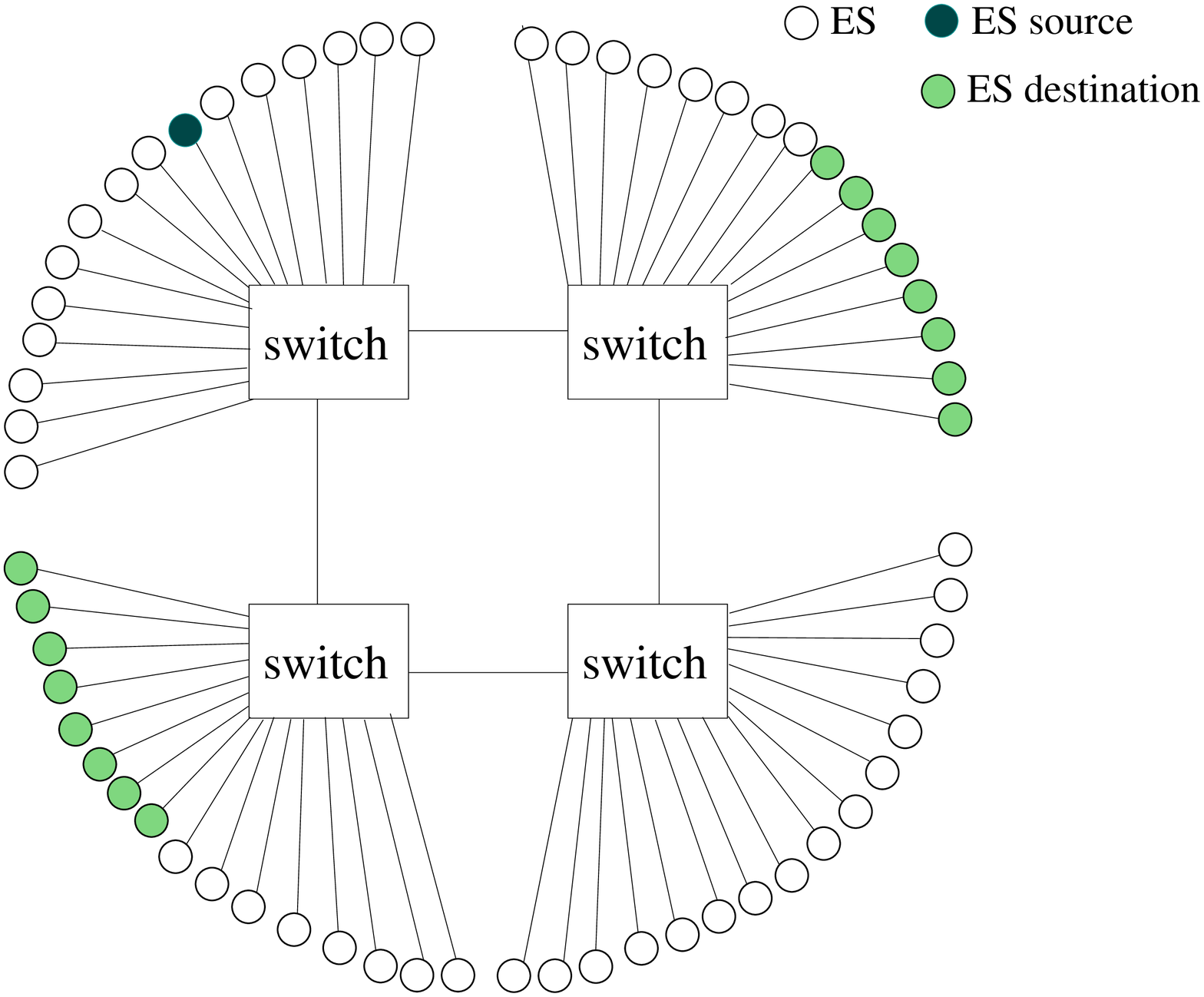}}
		\footnotesize \caption{Representative AFDX network: (a) Architecture; (b) Traffic communication patterns}
		\label{fig:rlfullnetwork}
	\end{figure}

	\begin{table}[h!]
		\footnotesize
		\begin{center}
			\begin{tabular}{|c|c|c|}
				\hline
				Scenarios &  Scenario 3& Scenario 4  \\
				\hline
				$(UR^{bn}_{SCT};UR^{bn}_{RC}) (\%)$ & $([0.4.. 45 ];20)$ & $(20;[2.. 80])$\\
				\hline
				$(n_{SCT}^{es};n_{RC}^{es})$&  $([1:4:110];10)$ & $(47;[1:2:39])$    \\
				\hline
				$(BW;L_M;L_R)$ & $(0.46;22,118;0)$ & $(0.46;22,118;0)$  \\
				\hline
			\end{tabular}
		\end{center}
		\footnotesize \caption{Considered Test Scenarios 3 and 4}
		\label{table34}
	\end{table}
	
	\subsection{Analysing the scalability of the current avionics configuration}
	\label{scalability}

		\begin{table}[h!]
			\footnotesize
			\begin{center}
				\begin{tabular}{|c|c|c|c|}
					\hline
					Configurations  & Legacy reference & Scenario 1& Scenario 2 \\
					\hline
					$UR^{bn}_{SCT}(\%)$& $28.7$ & $43$ & $28.7$ \\
					\hline
					$UR^{bn}_{RC} (\%)$& $3$ & $3$ & $13$ \\
					\hline
					$(L_M;L_R)$ & N.A& $(10,240;0)$  & $(35,840;0)$ \\
					\hline
					$BW$ & N.A& $0.90$  & $0.65$ \\
					\hline				
				\end{tabular}
			\end{center}
			\footnotesize \caption{Results of Scenarios 1 and 2}
			\label{res12}
		\end{table}
	
	The aim of this section is to analyse the scalability of our proposed Extended AFDX, in comparison with the legacy AFDX. Hence, starting from the \textit{legacy reference configuration} characterised by ($UR^{bn}_{RC}=3\%$ , $UR^{bn}_{SCT}=28.7\%$), we have tested for scenario 1 (resp. scenario 2) various BLS parameters to increase as much as possible the maximum SCT (resp. RC) utilisation rate. As shown in Table \ref{res12}, there exists a BLS configuration for scenario 1 (resp. scenario 2) allowing to achieve an utilisation rate of SCT (resp. RC) of  $UR^{bn}_{SCT} = 43\%$ ( resp. of $UR^{bn}_{RC}= 13\%$) under Extended AFDX, instead of only $UR^{bn}_{SCT} = 28.7\%$ ($UR^{bn}_{RC} = 3\%$) under legacy AFDX.
	
	\textbf {\textit { These results show an enhancement of scalability with the Extended AFDX of 50\% and 333\% for the SCT and RC classes, respectively, in comparison with the legacy AFDX.}}
	
	\subsection{Analysing timing performance of future avionics configurations}
	\label{timing}
	\begin{figure}[htbp]
		\centering
		\subfigure[]{\includegraphics[width=0.22\textwidth, angle=270]{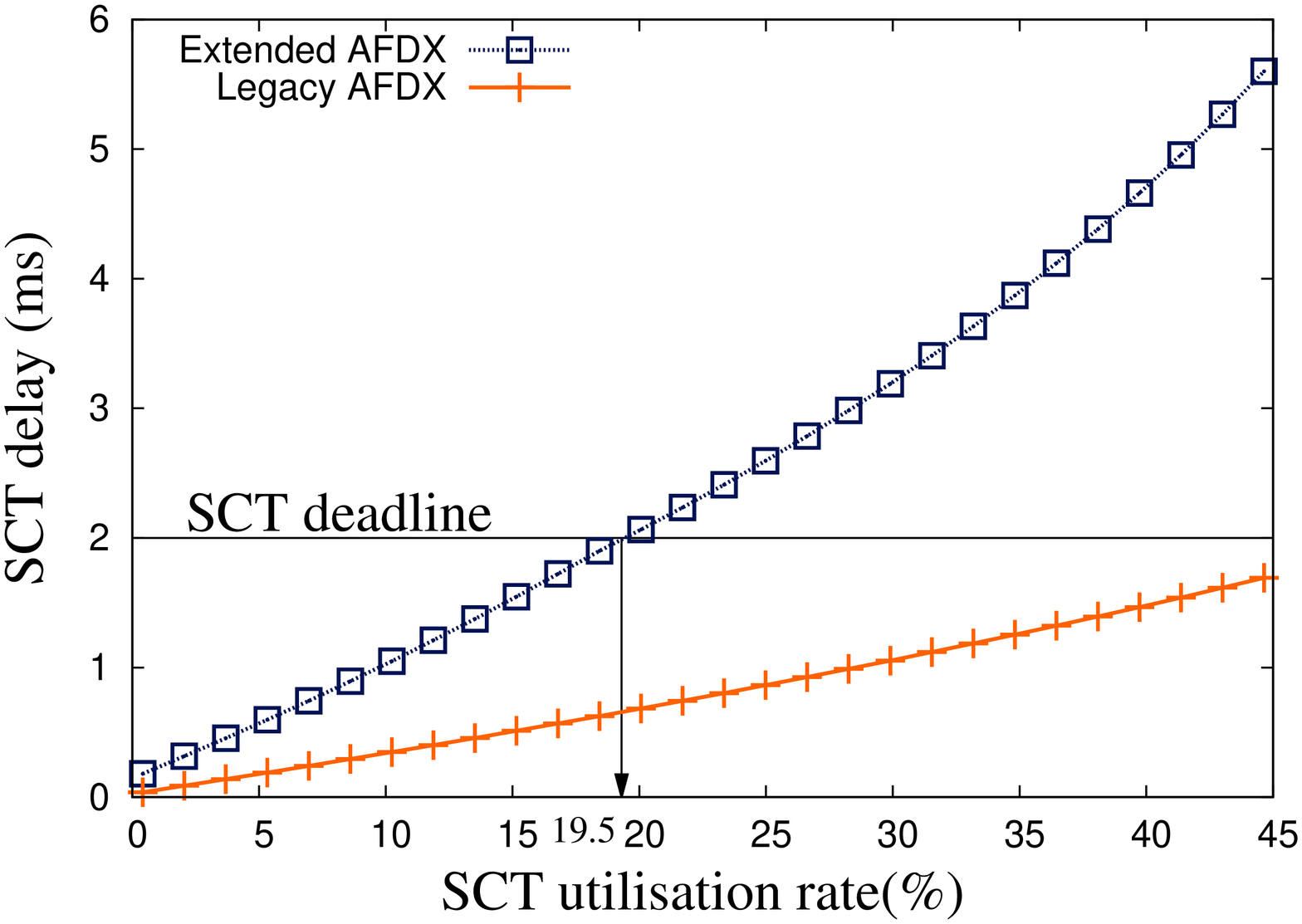}}
		\centering
		\subfigure[]{\includegraphics[width=0.22\textwidth, angle=270]{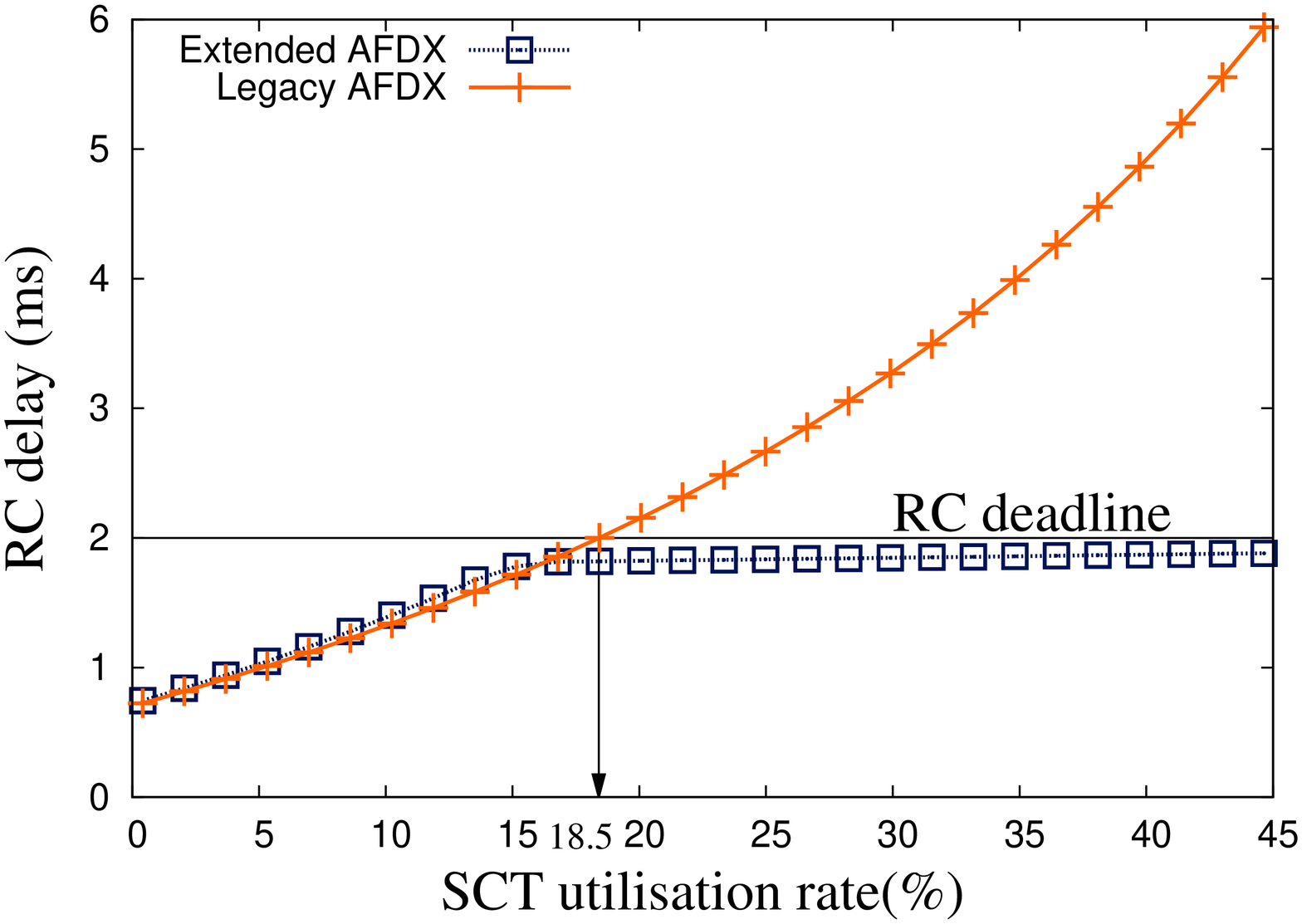}}
		\footnotesize \caption{Scenario 3: Impact of SCT max. utilisation rate on: (a) SCT delays; (b) RC delays}
		\label{fig:SCTImpact}
	\end{figure}

	\textbf{SCT timing performance}\\
	To analyse the timing performance of SCT when using the Extended AFDX instead of the legacy AFDX, we focus on Figures \ref{fig:SCTImpact}(a) and \ref{fig:AFDXImpact}(a) showing the SCT delay bounds evolution, regarding the SCT and RC bottleneck utilisation rate variation, respectively.
	
	As shown in Fig. \ref{fig:SCTImpact} (a), when increasing the bottleneck utilisation rate of the SCT, the delay upper bounds are obviously increasing under both solutions, but are globally higher under the Extended AFDX. This fact is due to the BLS behaviour on top of the NP-SP scheduler implemented within the Extended AFDX, which infers dividing the SCT burst to be sent within many \textit{sending windows}; whereas the regular NP-SP scheduler implemented within the legacy AFDX is sending the SCT burst all at once. 
	
	On the other hand, as it can be noticed in Fig. \ref{fig:AFDXImpact} (a), when increasing the bottleneck utilisation rate of the RC, the SCT delay bounds are constant under the legacy AFDX since SCT has the highest priority level and is at most delayed by a maximum sized frame of lower priorities, i.e., RC and BE; whereas they are increasing under the Extended AFDX for a RC utilisation rate up to 20\% and become equal to the SCT deadline (2ms) for a RC utilisation rate higher than 20\%. The increase is due to the fact that the RC rate is not large enough to use all the bandwidth guaranteed by the BLS; thus the guaranteed SCT service within the BLS is limited by the left part of the Eq.(\ref{eq:swSCT}). This shows the good isolation level, enforced by the BLS, between the mixed-criticality traffic, i.e., RC and SCT. 	
	
	\textbf {\textit {These results show the impact of the Extended AFDX network on SCT when increasing the network congestion. The main interesting feature to highlight is its efficiency to guarantee a high isolation level between the mixed criticality traffic, which is one of the key requirements for avionics applications.}}\\
	
	\begin{figure}[htbp]
		\centering
		\subfigure[]{\includegraphics[width=0.23\textwidth, angle=270]{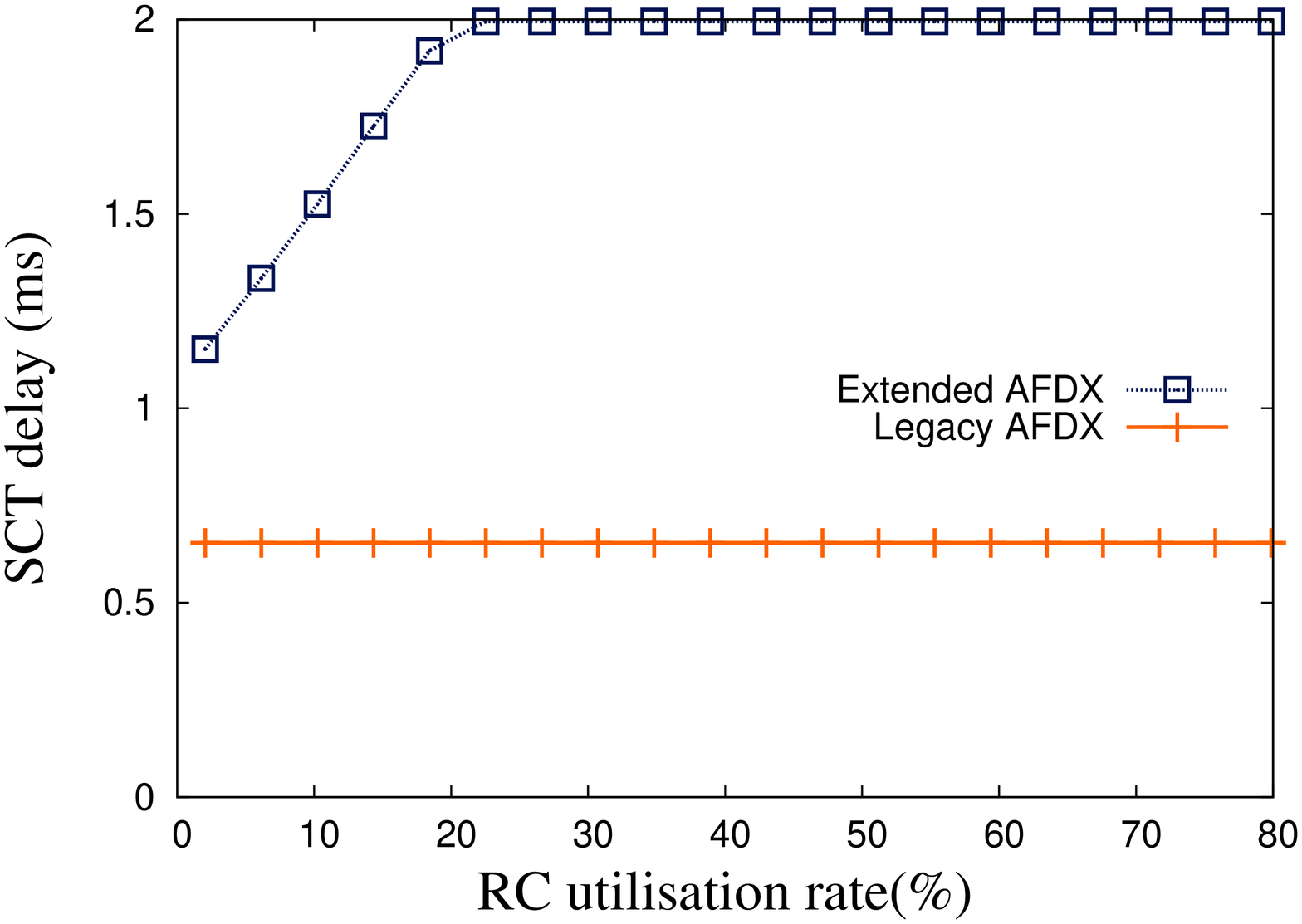}}
		\centering
		\subfigure[]{\includegraphics[width=0.23\textwidth, angle=270]{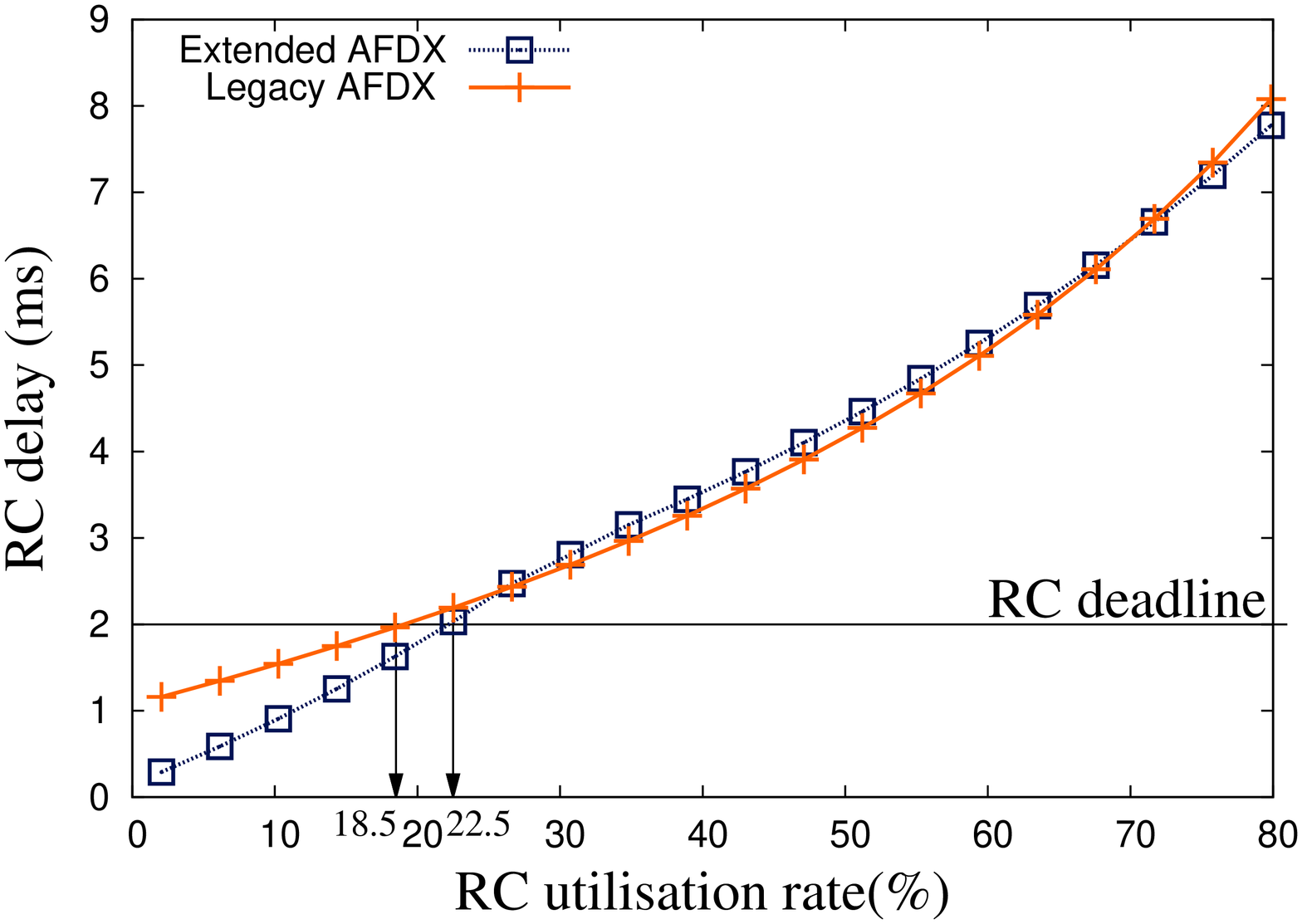}}
		\footnotesize \caption{Scenario 4: Impact of RC max. utilisation rate on: (a) SCT delays; (b) RC delays}
		\label{fig:AFDXImpact}
	\end{figure}
	\textbf{RC timing performance}\\
	We detail herein the main interesting results concerning the impact of our proposed solution on the RC timing performance, based on Figures \ref{fig:SCTImpact}(b) and \ref{fig:AFDXImpact}(b).
	
	Fig. \ref{fig:SCTImpact} (b) illustrates the variation of the RC delay upper bounds in terms of the SCT bottleneck utilisation rate. We can easily distinguish two phases on this figure. The first one is observed for a maximum utilisation rate below 14\%, where the delay upper bounds under both solutions are very similar. The second phase, i.e., when the maximum utilisation rate is higher than 14\%, shows that the delay upper bounds increase inherently under legacy AFDX, whereas they are constant under the Extended AFDX. 
	
	These results are coherent with the guaranteed service to the RC traffic in Th. \ref{th:bls-min-servicel}, which is the maximum between $\beta_{RC}^{sp}$ and $\beta_{RC}^{bls}$. Hence, during the first phase, the service corresponds to $\beta_{RC}^{sp}$, which is impacted by the maximum arrival curve of the SCT; thus its maximum utilisation rate. This fact explains the delay bounds increase. Afterwards, the service becomes related to $\beta_{RC}^{bls}$ during the second phase, which enforces a maximum constraint on the arrival curve of the SCT under the Extended AFDX due to the BLS, $\gamma_{SCT}^{bls}$. This maximum constraint implies a constant delay under the Extended AFDX. On the other hand, the service guaranteed under legacy AFDX is deeply related to the arrival curve of SCT, which explains the inherent delay bound increase. 
	
	Fig. \ref{fig:AFDXImpact} (b) shows the impact of the RC utilisation rate variation on the RC delay bounds. As it can be noticed, the RC delay bounds are increasing under both solutions, but still are better under Extended AFDX. For instance, for a RC utilisation rate of 10\%, we observe a delay bound of $1.5$ms and $0.9$ms under the legacy and Extended AFDX, respectively; thus the enhancement of the delay bound is about 40\% at $UR_{RC}^{bn}=10\%$, and it goes up to 74\% at $UR_{RC}^{bn}=2\%$.
	
	We need also to highlight the scalability enhancement in terms of utilisation rate of SCT and RC traffic under the Extended AFDX, in comparison to the legacy AFDX, which is coherent with the conclusions of the Section \ref{scalability}. In scenario 3 (resp. scenario 4), the maximum SCT (resp. RC) utilisation rate, respecting all the constraints, is 19.5\% (resp. 22.5\%) and 18.5\% (resp. 18.5\%) with Extended and legacy AFDX, respectively. 
	
	\textbf {\textit{These results show the valuable impact of the Extended AFDX on RC traffic, in comparison with the legacy AFDX solution. We can distinguish two interesting features: (i) the first one concerns the noticeable RC delay bounds decrease, where they become constant after a given SCT utilisation rate; (ii) the second one is the enhancement of the RC delay bound under the Extended AFDX when varying the RC utilisation rate, which is up to 74\%.}}

	\section{Conclusions}
	In this paper, we have defined and analysed the timing performance of an Extended AFDX network, to handle mixed criticality avionics applications. The proposed solution implements a BLS shaper on top of NP-SP within AFDX switches, to manage three priority levels. The conducted performance analysis on a realistic avionics case study highlights the benefit of using such a proposal, to isolate the highest priority and mitigate the impact of highest priority traffic on the RC one. Numerical results have shown noticeable enhancements of the delay upper bounds of the RC traffic (up to 74\%), and a gain in terms of maximum utilisation rate up to 333\% for RC and 50\% for SCT, in comparison with the legacy AFDX network.
	As a next step, we will introduce a tuning method to find the best BLS parameters, which respect the highest priority traffic deadline, while decreasing as much as possible the RC delay bounds. 

	\bibliographystyle{plain}
	\bibliography{biblio} 
	
\end{document}